\documentclass[aps,pra,reprint]{revtex4-1}
\usepackage{blindtext}
\usepackage[T1]{fontenc}
\usepackage[latin9]{inputenc}
\usepackage{amsthm}
\usepackage{amsmath}
\usepackage{amssymb}
\usepackage{color}
\usepackage{hyperref}
\usepackage{calc}
\usepackage{graphicx}
\usepackage{epstopdf}
\usepackage{mathtools}
\newtheorem{theorem}{Theorem}
\newtheorem{definition}{Definition}

\newtheorem{ex}{Example}
\newtheorem{fact}{Fact}
\newtheorem{lemma}{Lemma}

\DeclarePairedDelimiter\ceil{\lceil}{\rceil}


\begin{document}
\title{Criteria for universality of quantum gates}

\author{Adam Sawicki$^{1}$ and Katarzyna Karnas$^1$}
\email[E-mail: ]{a.sawicki@cft.edu.pl, karnas@cft.edu.pl}
\affiliation{$^1$Center for Theoretical Physics PAS, Al. Lotnik\'ow 32/46, 02-668, Warsaw, Poland}

\begin{abstract}
We consider the problem of deciding if a set of quantum one-qudit gates $\mathcal{S}=\{U_1,\ldots,U_n\}$ is universal. We provide the compact form criteria leading to a simple algorithm that allows deciding universality of any given set of gates in a finite number of steps. Moreover, for a non-universal $\mathcal{S}$ our criteria indicate what type of gates can be added to $\mathcal{S}$ to turn it into a universal set.
\end{abstract}
%

\maketitle

Universal quantum gates play an important role in quantum computing and quantum optics \cite{Loyd,Dojcz,Reck}. The ability to effectively manufacture gates operating on many modes, using for example optical networks that couple modes of light \cite{exp1,exp2}, is a natural motivation to consider the universality problems not only for qubits but also for higher dimensional systems, i.e. qudits (see also \cite{oszman,oszman2} for fermionic linear optics and quantum metrology). For quantum computing with qudits, a universal set of gates consists of all one-qudit gates together with an additional two-qudit gate that does not map separable states onto separable states \cite{Brylinski} (see \cite{daniel, zeier1, zeier, Laura} for recent results in the context of universal Hamiltonians). The set of all one-qudit gates can be, however, generated using a finite number of gates \cite{kuranishi}. We say that one-qudit gates $\mathcal{S}=\{U_1,\ldots,U_n\}\subset SU(d)$ are universal if any gate from $SU(d)$ can be built, with an arbitrary precision, using gates from $\mathcal{S}$. It is known that almost all sets of qudit gates are universal, i.e. non-universal sets $\mathcal{S}$ of the given cardinality are of measure zero and can be characterised by vanishing of a finite number of polynomials in the gates entries and their conjugates \cite{kuranishi,field}. Surprisingly, however, these polynomials are not known and it is hard to find operationally simple criteria that decide one-qudit gates universality. Some special cases of optical $3$-mode gates have been recently studied in \cite{BA,Sawicki} and the approach providing an algorithm for deciding universality of a given set of quantum gates that can be implemented on a quantum automata has been proposed \cite{Derksen05} (see also \cite{babai1,babai2,babai3} for algorithms deciding if a finitely generated group is infinite). The main obstruction in the problems considered in \cite{BA,Sawicki}  is the lack of classification of finite subgroups of $SU(d)$ for $d>4$. Nevertheless, as we show in this paper one can still provide some reasonable conditions for universality of one-qudit gates  without this knowledge. 

The efficiency of universal sets is typically measured by the number of gates that are needed to approximate other gates with a given precision $\epsilon$. The Solovay-Kitaev theorem states that all universal sets are roughly the same efficient. More precisely, the number of gates needed to approximate any gate $U\in SU(d)$ is bounded by $O(\log^c(1/\epsilon))$ \cite{NC00}, where $c$ may depend only on $d$ and $c\geq 1$. Recently there has been a bit of flurry in the area of single qubit gates \cite{bocharov,kliuchnikov,klucz,selinger} showing that using some number theoretic results and conjectures one can construct universal sets with $c = 1$. The approach presented in these contributions has been unified in \cite{sarnak} where the author pointed out the connection of these new results with the seminal work about distributing points on the sphere $S^2$ \cite{lubotzky} that uses results concerning optimal spectral gap for the averaging operator. Moreover, the authors of \cite{harrow} showed that the existence of the spectral gap implies  $c=1$ for all symmetric universal sets of single qudit gates, where by symmetric we mean the set $\mathcal{S}=\{U_1,\ldots, U_n\}$ with $n=2k$ and $U_{k+i}=U_i^{-1}$ for $i\in\{1,\ldots,k\}$. Although there are still some problems to solve in this area it seems that any further progress would require development of more advanced methods in pure mathematics rather than in quantum information. These developments should include verification of the spectral gap conjecture. Currently it is known to be true under the additional assumption that gates have algebraic entries \cite{BG1,BG2}. 

In this paper we present an approach that allows to decide universality of $\mathcal{S}$ by checking the spectra of the gates and solving some linear equations whose coefficients are polynomial in the entries of the gates and their complex conjugates. Moreover, for non-universal $\mathcal{S}$, our method indicates what type of gates can be added to make $\mathcal{S}$ universal. The paper is organised as follows. We start from presenting basic facts concerning the adjoint representation of $SU(d)$. The adjoint representations assigns to every matrix $U\in SU(d)$ a matrix $\mathrm{Ad}_U\in SO(d^2-1)$. We give the explicit formula for $\mathrm{Ad}_U$. The necessary condition for universality (Lemmas \ref{thm1} and \ref{lema2}) is then formulated using matrices $\mathrm{Ad}_U$ and $\mathrm{Ad}_{U^{-1}}$, where $U\in\mathcal{S}$ and boils down to checking the dimension of the kernel of the matrix $M_\mathcal{S}$ given by (\ref{matrix}). Next, we assume that the necessary condition for universality is satisfied and provide sufficient conditions for $<\mathcal{S}>$ to be infinite and thus dense in $SU(d)$. More precisely, if $<S>$ contains at least one element whose Hilbert-Schmidt distance from $Z(SU(d))=\{\alpha I:\alpha^d=1\}$ is both nonzero and less than $1/\sqrt{2}$ then $<\mathcal{S}>$ is infinite. Combining this with basic results in number theory we arrive at our main results. In Theorem \ref{main} we state that $\mathcal{S}$ is universal if $\mathcal{S}$ contains at least one matrix whose spectrum does not belong to some finite list of {\it exceptional spectra}. We also provide the algorithm which allows deciding universality of any given set of gates $\mathcal{S}$ (also when $\mathcal{S}$ contains matrices with exceptional spectra) in a finite number of steps. We discuss the correctness of the algorithm and provide instructive examples for $\mathcal{S}\subset SU(2)$.

\section{The necessary condition for universality}
Let us begin with introducing the basic notation used in this paper and explaining the adjoint representation. The set of gates $\mathcal{S}=\{U_1,\ldots,U_n\}\subset SU(d)$ is called universal if the set  generated by elements of $\mathcal{S}$
\begin{gather}
<\mathcal{S}>:=\{U_{i_1}\cdot\ldots\cdot U_{i_m}:U_{i_j}\in \mathcal{S}, \, m\in \mathbb{N}\}
\end{gather}
is dense in $SU(d)$, i.e. the closure $\overline{<\mathcal{S}>}=SU(d)$. In fact $\overline{<\mathcal{S}>}$ is always a Lie group \cite{KS16}. If this this group is $SU(d)$ we will say that $\mathcal{S}$ generates $SU(d)$.  

Let us denote by  $\mathfrak{su}(d)$ the Lie algebra of $SU(d)$. Recall that  $X\in\mathfrak{su}(d)$ iff  $X$ is an antihermitian traceless matrix. Moreover, the Lie algebra $\mathfrak{su}(d)$ is a real vector space equipped with a nondegenerate positive inner product defined by $(X|Y)=-\frac{1}{2}\mathrm{tr}XY$. For $U\in SU(d)$ and $X\in\mathfrak{su}(d)$ we define 
\[
\mathrm{Ad}_UX:=UXU^{-1}.
\]
One easily checks that $\mathrm{Ad}_U$ is a linear operator acting on $\mathfrak{su}(d)$. It is also invertible as $(\mathrm{Ad}_U)^{-1}=\mathrm{Ad}_{U^{-1}}$ and preserves the inner product as 
\begin{gather}
(\mathrm{Ad}_UX|\mathrm{Ad}_UY)=-\frac{1}{2}\mathrm{tr}UXU^{-1}UYU^{-1}=\nonumber\\=-\frac{1}{2}\mathrm{tr}XY=(X|Y)
\end{gather}
 Therefore $\mathrm{Ad}_U$ is an orthogonal transformation acting on $d^2-1$ dimensional vector space $\mathfrak{su}(d)$.  Upon a choice of an orthonormal basis $\{X_i\}_{i=1}^{d^2-1}$ in $\mathfrak{su}(d)$, i.e. basis that statisfies $(X_i|X_j)=\delta_{ij}$ the transformation $\mathrm{Ad}_U$ can be expressed in this basis as a matrix belonging to $SO(d^2-1)$, i.e $\mathrm{Ad}_U^t\mathrm{Ad}_U=I$ and $\det\mathrm{Ad}_U=1$. The entries of this matrix, $(\mathrm{Ad}_U)_{ij}$, are real and  defined by the identity:
\begin{gather}
\mathrm{Ad}_UX_j=U^{-1}X_jU=\sum_{i=1}^d(\mathrm{Ad}_U)_{ij}X_{i},
\end{gather} 
thus they are give by 
\begin{gather}\label{ad}
(\mathrm{Ad}_U)_{ij}=-\frac{1}{2}\mathrm{tr}\left(X_iUX_jU^{-1}\right).
\end{gather} 

Note that we also have $\mathrm{Ad}_{U_1U_2}=\mathrm{Ad}_{U_1}\mathrm{Ad}_{U_2}$ and this way we obtain the homomorphism 
\begin{gather}
\mathrm{Ad:}\;SU(d)\to SO(d^2-1),
\end{gather}
that is known as the {\it adjoint representation}.

For a set of $d\times d$ real matrices $M$, let us denote the set of all $d\times d$  matrices commuting with matrices from $M$ by
\begin{gather}
\mathcal{C}(M)=\{L:[L,m]=0,\,\forall m\in M\}.
\end{gather}
The adjoint representation of $SU(d)$ is an absolutely irreducible real representation and therefore by the extended version of Schur's lemma \cite{Dieck,hall}, the only $(d^2-1)\times (d^2-1)$ matrix that commutes with all matrices $\mathrm{Ad}_{SU(d)}=\{\mathrm{Ad}_{U}:U\in SU(d)\}$ is proportional to the identity matrix, $I$. In other words $\mathcal{C}({\mathrm{Ad}_{SU(d)}})=\{\lambda I:\lambda \in\mathbb{R}\}$. 
 
 \begin{ex}
The adjoint representation for $d=2$, i.e. $\mathrm{Ad}:SU(2)\rightarrow SO(3)$ has a particularly nice form. Any matrix from $SU(2)$ can be written in a form
\begin{gather}
U(\phi,\vec{k})=I\cos\phi+\sin\phi(k_xX+k_yY+k_zZ),
\end{gather}
where $Y=i\sigma_1,\;X=i\sigma_2,\;Z=i\sigma_3$ and $\sigma_i$ are Pauli matrices, $\vec{k}=(k_x,k_y,k_z)^T$ satisfies $k_x^2+k_y^2+k_z^2=1$.
Similarly, any matrix from $SO(3)$ has a form
\begin{gather}
O(\phi,\vec{k}) = I+\sin\phi(-k_xX_{12}+k_yX_{13}-k_zX_{23})+\label{def:so3_element}\\+2\sin^2\frac{\phi}{2}(-k_xX_{12}+k_yX_{13}-k_zX_{23})^2,\nonumber
\end{gather}
where $X_{ij}=E_{ij}-E_{ji}$, and $E_{ij}$ is a matrix whose only non vanishing entry is $(i,j)$. One easily verifies that the adjoint representation is given by  
\begin{gather}
\mathrm{Ad}_{U(\phi,\vec{k})}= O(2\phi,\vec{k}).
\end{gather}
 \end{ex}
 For $U\in SU(d)$, where $d>2$ calculation of matrices $\mathrm{Ad}_U$ can be done using formula (\ref{ad}) upon the choice of orthonormal basis in $\mathfrak{su}(d)$. For $d=3$ this basis is given by, for example, the Gell-Mann matrices multiplied by imaginary unit $i$. For higher $d$ one can construct an orthonormal basis of $\mathfrak{su}(d)$ in an analogous way as for $d=3$. 

General considerations that can be found in \cite{KS16} show that the group $\overline{<S>}$ can be either:
\begin{enumerate}
\item $\overline{<S>}=SU(d)$, or
\item $\overline{<S>}$ is infinite and connected, or
\item $\overline{<S>}$ is infinite and consists of $k<\infty$ connected components, where each component has the same dimension (as a manifold), or
\item $\overline{<S>}$ is finite.
\end{enumerate}
Note that in cases 1, 2, and 3 the group $\overline{<S>}$ has infinite number of elements. Thus we first provide criteria that distinguish between case 1 and  cases 2 and 3. To this end we will use the adjoint representation. For $\mathcal{S}=\{U_1,\ldots, U_n\}\subset SU(d)$ let $\mathrm{Ad}_{\mathcal{S}}=\{\mathrm{Ad}_{U}:U\in\mathcal{S}\}$. Note that if $[L,\mathrm{Ad}_{U_1}]=0$ and $[L,\mathrm{Ad}_{U_2}]=0$  then 
\begin{gather*}
[L,\mathrm{Ad}_{U_1U_2}]=[L,\mathrm{Ad}_{U_1}]\mathrm{Ad}_{U_2}+\mathrm{Ad}_{U_1}[L,\mathrm{Ad}_{U_2}]=0.
\end{gather*}
Thus if $\mathcal{S}$ generates $SU(d)$ and $L$ is a matrix that commutes with $\mathrm{Ad}_\mathcal{S}$ then $L$ commutes with $\mathrm{Ad}_{SU(d)}$. Therefore for universal $\mathcal{S}$ we have $\mathcal{C}(\mathrm{Ad}_{\mathcal{S}})=C({\mathrm{Ad}_{SU(d)}})=\{\lambda I:\lambda \in \mathbb{R}\}$. It turns out (see \cite{KS16}) that the converse is true under one additional assumption, namely that $<\mathcal{S}>$ is infinite.
\begin{lemma}\label{thm1}
For a set of special unitary matrices $\mathcal{S}=\{U_1,\ldots,U_n\}$ assume that ${<\mathcal{S}>}$ is infinite and $\mathcal{C}(\mathrm{Ad}_{\mathcal{S}})=\{\lambda I:\lambda \in \mathbb{R}\}$. Then $\overline{<\mathcal{S}>}=SU(d)$.
\end{lemma}
The proof of this lemma is based on the structure theory for semisimple Lie groups and can be found in \cite{KS16}. Here we only make some additional remarks regarding calculation of $\mathcal{C}(\mathrm{Ad}_{\mathcal{S}})$. Let $\mathrm{vec}(L)$ be the vectorisation of matrix $L$, i.e. the vector obtained by stacking the columns of the matrix $L$ on top of one another. One easily calculates that 
\begin{gather*}
[L,\mathrm{Ad}_U]=0\Leftrightarrow \left(I\otimes \mathrm{Ad}_U-\mathrm{Ad}_{U^\dagger}\otimes I\right)\mathrm{vec}(L)=0,
\end{gather*}
where $U^\dagger$ is the complex conjugate and transpose of $U$, i.e. $U^\dagger=\bar{U}^t$. Let 
\begin{gather}\label{matrix}
M_{\mathcal{S}}=\left(\begin{array}{c}
I\otimes \mathrm{Ad}_{U_1}-\mathrm{Ad}_{U_1^\dagger}\otimes I \\
\vdots\\
I\otimes \mathrm{Ad}_{U_n}-\mathrm{Ad}_{U_n^\dagger}\otimes I \\
\end{array}\right)
\end{gather} 
\begin{lemma}\label{lema2}
$\mathcal{C}(\mathrm{Ad}_{\mathcal{S}})=\{\lambda I:\lambda \in \mathbb{R}\}$ if and only if the kernel of $M_{\mathcal{S}}$ is one-dimensional.
\end{lemma}
We emphasise  the role of the adjoint representation which is crucial in Lemma \ref{thm1}. In particular there are infinite subgroups $\overline{<\mathcal{S}>}$ such that $\mathcal{C}(\mathcal{S})=\mathcal{C}({SU(d)})$ but $\overline{<\mathcal{S}>} \neq SU(d)$. In Example 1 we provide such a subgroup for $d=2$. We next characterise space $\mathcal{C}(\mathrm{Ad}_\mathcal{S})$ for $SU(2)$.

Let us recall that the composition of two unitary matrices $U(\gamma,\vec{k}_{12})=U(\phi_1,\vec{k}_1)U(\phi_2,\vec{k}_2)$ is a unitary matrix with $\gamma$ and $\vec{k}_{12}$ determined by:
\begin{small}
\begin{gather}
\cos\gamma = \cos\phi_1\cos\phi_2-\sin\phi_1\sin\phi_2\vec{k}_1\cdot\vec{k}_2,\label{def:gamma}\\\label{def:k12}
\vec{k}_{12}=\frac{1}{\sin\gamma}(\vec{k}_1\sin\phi_1\cos\phi_2+\vec{k}_2\sin\phi_2\cos\phi_1+\\\nonumber
+\vec{k}_1\times\vec{k}_2\sin\phi_1\sin\phi_2).
\end{gather}
\end{small}
Moreover, two unitary matrices $U_1(\phi_1,\vec{k}_1)$, $U_2(\phi_2,\vec{k}_2)$ commute iff $\vec{k}_1\parallel\vec{k}_2$ or $\phi=k\pi$. Similarly two orthogonal matrices $O_1(\phi_1,\vec{k}_1)$, $O_2(\phi_2,\vec{k}_2))$ commute if  $\vec{k}_1\parallel\vec{k}_2$ or  one of $\phi_i$'s is an even multiple of $\pi$, or $\phi_1=\pm\pi=\phi_2$ and $\vec{k}_1\perp\vec{k}_2$. Making use of these facts in \cite{KS16} we show:
\begin{fact}\label{fact1}
For noncommuting $U_1(\phi_1,\vec{k}_1)$, $U_2(\phi_2,\vec{k}_2)$, the space $\mathcal{C}(\mathrm{Ad}_{U_1(\phi_1,\vec{k}_1)},\mathrm{Ad}_{U_2(\phi_2,\vec{k}_2)})$ is larger than $\{\lambda I:\lambda\in\mathbb{R}\}$ if and only if: (1) $\phi_1=\frac{k\pi}{2}=\phi_2$, (2) one of $\phi_i$'s is equal to $\frac{k\pi}{2}$ and $\vec{k}_1\perp\vec{k}_2$, where $k$ is an odd integer.
\end{fact}

\section{When is $<S>$ infinite?}

We next describe the conditions under which $\overline{<\mathcal{S}>}$ is infinite. For $U_1,U_2\in SU(d)$ the group commutator is defined as $[U_1,U_2]_\bullet = U_1U_2U_1^{-1}U_2^{-1}$. Note that $[U_1,U_2]=0$ is equivalent to $[U_1,U_2]_\bullet =I$. The distance between elements of $SU(d)$ can be measured using the Hilbert-Schmidt norm  $||U|| = \sqrt{\mathrm{tr}UU^\dagger}$. For two elements $U_1,U_2$ we have the following relation between their distances from the identity and the distance of their group commutator from the identity \cite{curtis}: 
\begin{gather}
||[U_1,U_2]_\bullet- I|| \leq \sqrt{2}||U_1- I||\cdot ||U_2- I||,\label{def:inequality}\\
[U_1,[U_1,U_2]_\bullet]_\bullet=I\:\mathrm{and}\:||U_2-I||<2\Rightarrow [U_1,U_2]_\bullet=I.\nonumber
\end{gather} 
Let 
\begin{gather}
B_{\alpha}=\{U\in SU(d):\|U-\alpha I\|\leq \frac{1}{\sqrt{2}}\}\subset SU(d),
\end{gather} 
be a ball of radius $\frac{1}{\sqrt{2}}$ that is centred at elements $\alpha I$. As $\det\alpha I=1$ we need to assume $\alpha^d=1$. Let 
\begin{gather}
\mathcal{B}=\bigcup_{\alpha^d=1} B_\alpha.
\end{gather}
It turns out that noncommuting elements belonging to $\mathcal{B}$ generate infinite subgroups of $SU(d)$:
\begin{lemma}
Assume that $[U_1,U_2]_\bullet\notin Z(SU(d))$ and $U_1,U_2\in \mathcal{B}$. Then $<U_1,U_2>$ is infinite. 
\label{lem:balls}
\end{lemma}
One of the steps in the proof of lemma \ref{lem:balls} uses relations (\ref{def:inequality}) to show that the sequence $g_0=U_1$, $g_1=[U_1,U_2]_\bullet$, $g_k=[g_{k-1},U_2]_\bullet$ converges to $I$ and $g_n\neq I$ for any integer $n$ \cite{field} (see \cite{KS16} for the full discussion).

We next describe when $U\in B_{\alpha_m }$, where $\alpha_m=e^{i \theta_m}$ and $\theta_m=\frac{2m}{d}\pi$. To this end note that 
\begin{gather}\label{trdist}
\|U-\alpha_m I\|^2=2\mathrm{tr}I-\alpha_m \mathrm{tr}U^\dagger-\alpha_m^\ast \mathrm{tr}U.
\end{gather}
 As the trace of $U$ is determined by its spectrum, the desired condition can be expressed in terms  of the eigenvalues of $U$ that are given by $\{e^{i\phi_1},...,e^{i\phi_d}\}$, $\phi_i\in [0,2\pi[$  and $\sum_{i=1}^{d}\phi_{i}=0\,\mathrm{mod}\,2\pi$. Direct calculations lead to:
\begin{gather}\label{condB1}
U\in B_{\alpha_m} \Leftrightarrow\sum_{i=1}^{d}\sin^{2}\frac{\phi_{i}-\theta_m}{2}<\frac{1}{8}.
\end{gather}

 Let us next assume that $U\in SU(d)$ does not belong to $\mathcal{B}$. Then one can show that there always exists an integer $n$ such that $U^{n}$ belongs to some $B_{\alpha}$ \cite{KS16}, $\alpha I\in Z(SU(d))$. For a given $U$, let $n_U$ be the smallest integer satisfying this condition. In \cite{KS16} we prove the modified version of the Dirichlet's approximation theorem and use it to find an upper bound for $N_{SU(d)}:=\mathrm{max}_{U}n_U$. This way, for every $U\in SU(d)$ there is $1\leq n\leq N_{SU(d)}$ such that $U^{n}\in B_{\alpha }$ for some $\alpha I \in Z(SU(d))$. Thus by taking powers $1\leq n\leq N_{SU(d)}$ we can move every element of $SU(d)$ into $\mathcal{B}$. 
 Assume next that the necessary condition for universality is satisfied, i.e. $\mathcal{C}(\mathrm{Ad}_{\mathcal{S}})=\{\lambda I: \lambda \in \mathbb{R}\}$. From Lemma \ref{thm1} one can easily deduce that under the assumption that $<\mathcal{S}>$ is infinite  the intersection $<\mathcal{S}>\cap \mathcal{B}$ is dense in $\mathcal{B}$. As we have shown in \cite{KS16} the necessary condition for universality places significant constrains on the structure of $<\mathcal{S}>$ also in the case when $<\mathcal{S}>$ is a finite group, namely we have that $<\mathcal{S}>\cap \mathcal{B}$ is a subset of $Z(SU(d))$. Thus $<\mathcal{S}>$ is finite if and only if there are no elements in $<\mathcal{S}>$ that belong to $\mathcal{B}$ other than those in $Z(SU(d))$. The above discussion is summarised by:
\begin{lemma}\label{term}
Let $\mathcal{S}=\{U_1,\ldots,U_n\}$ and assume that $\mathcal{C}(\mathrm{Ad}_{\mathcal{S}})=\{\lambda I:\lambda \in\mathbb{R}\}$. Then $\overline{<\mathcal{S}>}=SU(d)$ if and only if there is at least one matrix $U\in<\mathcal{S}>$ that belongs to $\mathcal{B}\setminus Z(SU(d))$.
\end{lemma}

We know that every element of $\mathcal{S}$ can be put to $\mathcal{B}$ by taking powers (that are bounded by $N_{SU(d)}$). Hence when $<\mathcal{S}>$ is finite introducing $U\in <\mathcal{S}>$ to $\mathcal{B}$ must be equivalent to introducing it to $Z(G)$. This  condition can be phrased in terms of specra of the matrices from $\mathcal{S}$.

\begin{definition}\label{exangle}
Assume $U\notin \mathcal{B}$. The spectrum of $U$ is called exceptional if it consists of $n^{\mathrm{th}}$ roots of $\alpha\in\mathbb{C}$ where $\alpha^d=1$ and $1\leq n\leq N_{SU(d)}$.
\end{definition}
The set of exceptional spectra is a finite set.

To illustrate the above ideas we find $N_{SU(d)}$ and the list of exceptional spectra for $d=2$. Note that for any $U\in SU(2)$ the spectrum is given by $\{e^{i\phi},e^{-i\phi}\}$ and therefore is determined by one angle $\phi$. The angle corresponding to an exceptional spectrum will be called an exceptional angle. Moreover, the centre of $SU(2)$ consists of two matrices $Z(SU(2))=\{I,-I\}$. We start with recalling the Dirichlet approximation theorem \cite{Dirichlet}. 
\begin{fact}[Dirichlet]\label{d1}
For a given real number $b$ and a positive integer $N$ there exist integers $1\leq n\leq N$ and $p$ such, that $nb$ differs from $p$ by at  most $\frac{1}{N+1}$, i.e.
\begin{gather}\label{dirichlet_one}
|nb-p|\leq\frac{1}{N+1}.
\end{gather}
\end{fact}
Let $[0,2\pi)\ni\phi=2b\pi$ be the spectral angle of $U$. By Fact \ref{d1} for a given $N$ there are integers $p$ and $1\leq n \leq N$ such that $|nb-p|\leq\frac{1}{N+1}$. Multiplying this inequality by $\frac{\pi}{2}$ we obtain $|n\frac{\phi}{2}-p\frac{\pi}{2}|\leq\frac{\pi}{2(N+1)}$. Note that  (\ref{condB1}) simplifies to $|\sin\frac{\psi}{2}|<\frac{1}{4}$ or $|\sin\frac{\psi-\pi}{2}|<\frac{1}{4}$. Thus for a given $\phi$ we search for $n$ satisfying  $|n\frac{\phi}{2}-p\frac{\pi}{2}|<\arcsin\frac{1}{4}$. Thus $\frac{\pi}{2(n+1)}<\arcsin\frac{1}{4}$ and 
\begin{gather}
n\leq\ceil*{\frac{\frac{\pi}{2}-\arcsin\frac{1}{4}}{\arcsin\frac{1}{4}}}=6.
\label{nmaxSu3}
\end{gather} 
The above upper bound for $N_{SU(d)}$ is attained for $\phi=\arcsin\frac{1}{4}$ (see figure \ref{condB1}). Hence $N_{SU(2)}=6$.
\begin{figure}[ht!]
\begin{center}\includegraphics[scale=0.44]{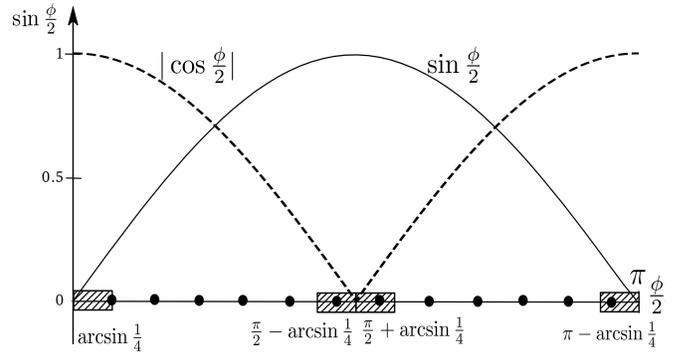}
\caption{Condition (\ref{condB1}) for $SU(2)$. The gray segments are determined by $|\sin\frac{\phi}{2}|<\frac{1}{4}$ and the white segments by $|\sin\frac{\phi-\pi}{2}|<\frac{1}{4}$.}
\end{center}
\end{figure}

Exceptional spectra for $SU(2)$, are determined by roots of $1$ or $-1$ of order $1\leq n\leq 6$, or equivalently, by primitive roots of unity of order $1\leq n\leq 6$ and $8$, $10$, $12$. More precisely, they are given by $\{e^{i\theta},e^{-i\theta}\}$, where 
\begin{gather}
\theta\in \{ k\pi,\frac{k_2\pi}{2},\frac{k_3\pi}{3},\frac{k_4\pi}{4},\frac{k_5\pi}{5},\frac{k_6\pi}{6} \},
\end{gather}
and $\gcd(k_i,i)=1$. The number of exceptional angles can be calculated using the Euler totient function,  and is equal $\sum_{i=1}^6\phi(i)+\sum_{i=4}^6\phi(2i)=24$. For higher dimensional groups, as we discuss in \cite{KS16},  the number $N_{SU(d)}$ grows exponentially with $d$. 
 Our main result is:
\begin{theorem}\label{main}
Assume that $\mathcal{C}(\mathrm{Ad}_{U_1},\ldots, \mathrm{Ad}_{U_n})=\{\lambda I:\lambda \in\mathbb{R}\}$ and at least one matrix $U_{i}$ has a nonexceptional spectrum. Then $\overline{<U_1,\ldots,U_n>}=SU(d)$. 
\end{theorem}
\section{The algorithm for deciding universality of $\mathcal{S}=\{U_1,\ldots,U_n\}\subset SU(d)$}
The case when all matrices $\{U_1,\ldots, U_n\}$ have exceptional spectra requires an algorithm which we next present. Our algorithm allows deciding universality of any given set of $SU(d)$ gates in a finite number of steps. 

\textbf{The algorithm}
\begin{description}
\item[Step 1]Check if $\mathcal{C}(\mathrm{Ad}_\mathcal{S})=\{\lambda I: \lambda\in\mathbb{R}\}$. This can be done by checking the dimension of the kernel of the matrix $M_\mathcal{S}$ (\ref{matrix}) constructed from the entries of matrices $\{\mathrm{Ad}_{U_1},\ldots, \mathrm{Ad}_{U_n}\}$ and thus is a linear algebra problem. If the answer is NO stop as the set $\mathcal{S}$ is not universal. If YES, set $l=1$ and go to step 2.
\item[Step 2] Check if there is a matrix $U\in \mathcal{S}$ for which $U^{n_U}$ belongs to $\mathcal{B}$ but not to $Z(SU(d))$, where $1\leq n_U\leq N_{SU(d)}$. This can be done using formula (\ref{trdist}). If the answer is YES $\mathcal{S}$ is universal. If the answer is NO,  set $l=l+1$.
\item[Step 3] Define the new set $\mathcal{S}$ by adding to $\mathcal{S}$ words of length $l$, i.e products of elements from $\mathcal{S}$ of length $l$. If the new $\mathcal{S}$ is equal to the old one, the group  $\overline{<\mathcal{S}>}$ is finite. Otherwise go to step 2.  
\end{description}

The major advantage of our approach is the fact that we can make decisions in steps 2 and 3 in finite `time'. It is also clear that for randomly chosen matrices $\mathcal{S}=\{U_1,\ldots, U_n\}\subset SU(d)$ our algorithm terminates with probability 1 in Step 2 for $l=1$. This is a direct consequence of the fact that exceptional spectra form a finite set.

Let us next discuss the correctness of our algorithm. Assume that $\mathcal{S}$ passes positively the necessary condition for universality, i.e. the Step 1. If the group $\overline{<\mathcal{S}>}$ is finite the algorithm terminates in Step 3 for some finite $l$. On the other hand, as a direct consequence of Lemma \ref{term}, for an infinite $<\mathcal{S}>$ the algorithm must terminate in Step 2 for some finite $l$. One can also argue that if all finite-length words have exceptional spectra then they cannot form a dense subset. Thus if $<\mathcal{S}>$ is dense then it must contain words of a finite length that have non-exceptional spectra. These words terminate the algorithm in Step 2. Moreover, we have the following: 
\begin{fact}\label{words}
Assume $<\mathcal{S}>$ is dense in $SU(d)$. The length of a word that gives termination of the universality algorithm is at most the length $l$ such that words of length $k\leq l$ form an $\epsilon$-net that covers $SU(d)$, where $\epsilon=\frac{1}{2\sqrt{2}+\delta}$ and $\delta>0$ is arbitrary small. 
\end{fact}
\begin{proof}
Assume that words of the length $k\leq l$ built from elements $\mathcal{S}$ form an $\epsilon$-net for $SU(d)$, where $\epsilon=\frac{1}{2\sqrt{2}+\delta}$ and $\delta>0$ is arbitrary small. Let $U$ be an element of $SU(d)$ whose distance from the identity is exactly $\frac{1}{2\sqrt{2}}$ (see Figure \ref{proof3}). Then by the definition of $\epsilon$-net there must be at least one word $w\in <\mathcal{S}>$ of length $k\leq l$ contained in the ball $C$ of radius $\epsilon=\frac{1}{2\sqrt{2}+\delta}$ centred at $U$. But this ball is contained in $B_1\setminus I$. Hence $w$ gives termination of the universality algorithm in Step 2. The result follows.
\end{proof}

\begin{figure}[ht!]
\begin{center}\includegraphics[
width=0.65\linewidth, height=0.6\linewidth]{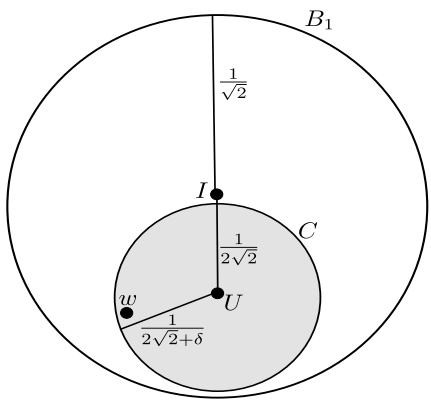}
\caption{\label{proof3} The proof of Fact \ref{words}.}
\end{center}
\end{figure}
The formulation of Fact \ref{words} is related to the results contained in \cite{freedman}.

In order to demonstrate how efficient is our algorithm we determine the maximal $l$ which gives its termination in Step 2 and Step 3 respectively for $SU(2)$. For simplicity we consider $\mathcal{S}$ of the form  $S=\{U(\phi_1,\vec{k}_1),U(\phi_2,\vec{k}_2)\}\subset SU(2)$. To this end it is enough to consider the case when both $\phi_1$ and $\phi_2$ are exceptional angles and the product $U(\phi_{12}, \vec{k}_{12})=U(\phi_1,\vec{k}_1)U(\phi_2,\vec{k}_2)$ has exceptional $\phi_{12}$ as otherwise the algorithm terminates in Step 2 with either $l=1$ or $l=2$. For each such pair (there are finitely many of them) we find the Step and $l$ that gives termination of our algorithm. The detailed discussion of the results and their connection to finite subgroups of $SU(2)$ can be found in \cite{KS16}. In short words, it turns out that our algorithm terminates at Step 2 if and only if $1\leq l\leq 4$ and at Step 3 when $5\leq l\leq 13$. 
\begin{fact}\label{algsu2}
For $\mathcal{S}=\{U(\phi_1,\vec{k}_1),U(\phi_2,\vec{k}_2)\}\subset SU(2)$ the algorithm for checking universality terminates for $l\leq 13$. Moreover,  the set $\mathcal{S}$ is universal if and only if the algorithm terminates for $l\leq 4$. 
\end{fact}

%

The main conclusion from Fact \ref{algsu2} is that one can decide universality of any two-element subset of $SU(2)$ by looking at words of the length at most $4$.

\section{Examples for $SU(2)$}

In the remaining part of this paper we demonstrate our approach calculating a few examples. They are chosen particularly to elucidate the importance of the conditions given by Theorem \ref{main}. 


\paragraph{\textbf{Example 1}} Let $\mathcal{S}=\{U(\phi,\vec{k}_1),U(\pi/2,\vec{k}_2)\}$, where $\vec{k}_1\perp\vec{k}_2$ and $\phi_1$ is an irrational multiple $\pi$. For example, when $\vec{k}_1=(0,0,1)$ and $\vec{k}_2=(1,0,0)$, we have
\begin{gather}
U(\phi,\vec{k}_1)=\left(\begin{array}{cc}
e^{-i\phi}& 0\\
0 & e^{i\phi} \\
\end{array}\right),\,U(\pi/2,\vec{k}_2)=\left(\begin{array}{cc}
0 & 1 \\
-1 & 0 \\
\end{array}\right).
\end{gather}
$U(\phi,\vec{k}_1)$ is of an infinite order and since $U(\phi,\vec{k}_1)$ and $U(\pi/2,\vec{k_2})$ do not commute we have that  $<\mathcal{S}>$ is infinite and not abelian. By Fact \ref{fact1}, however, 
\[
\mathcal{C}(\mathrm{Ad}_{U(\phi,\vec{k}_1)},\mathrm{Ad}_{U(\pi/2,\vec{k}_2)})\neq \{\lambda I:\lambda\in\mathbb{C}\},
\]
and hence $\overline{<\mathcal{S}>}\neq SU(2)$. For example $O(\pi,\vec{k}_1)\in SO(3)$ commutes with both $\mathrm{Ad}_{U(\phi_1,\vec{k}_1)}=O(2\phi_1,\vec{k}_1)$ and  $\mathrm{Ad}_{U(\pi/2,\vec{k}_2)}=O(\pi,\vec{k}_2)$. Interestingly, however 
\[
\mathcal{C}(U_1,U_2)=\{\lambda I:\lambda\in\mathbb{C}\}.
\] 
To understand the structure of the group $\overline{<\mathcal{S}>}$ note that $U(\pi/2,\vec{k_2})U(\phi_1,\vec{k}_1)U^{-1}(\pi/2,\vec{k_2})=U^{-1}(\phi_1,\vec{k}_1)$. Hence $U(\pi/2,\vec{k_2})$ is a normaliser of $<U(\phi_1,\vec{k}_1)>$. Thus the group $\overline{<\mathcal{S}>}$ consists of two connected components. The first one is given by one-parameter group $U(t,\vec{k}_1)$, where $t\in\mathbb{R}$ and the other one by elements of the form $U(\pi/2,\vec{k_2})U(t,\vec{k}_1)$. The adjoint representation is able to identify infinite disconnected subgroups whereas the defining representation is not. Moreover, we know exactly how to fix non-universality of the set $\mathcal{S}$. For example, we can add one matrix $U(\gamma,\vec{k}_\gamma)$ such that $\gamma\neq k\pi$ and $\vec{k}_\gamma$ is neither parallel nor orthogonal to $\vec{k}_1$ and $\vec{k}_2$. 
\begin{figure}[ht!]
\begin{center}\includegraphics[
width=\linewidth, height=0.7\linewidth]{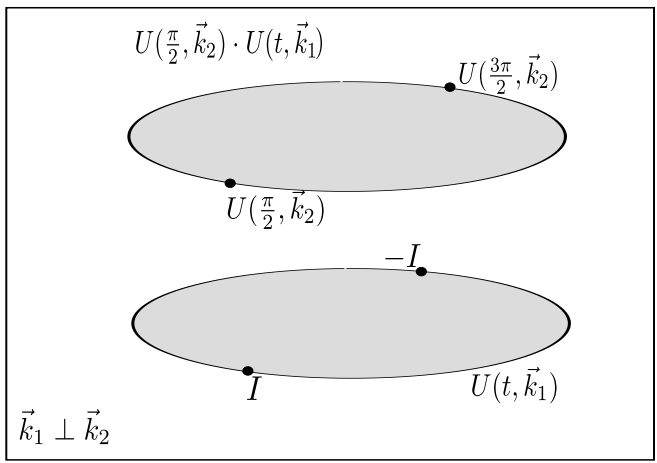}
\caption{Two component group generated by $\mathcal{S}=\{U(\phi,\vec{k}_1),U(\pi/2,\vec{k}_2)\}$.}
\end{center}
\end{figure}
\paragraph{\textbf{Example 2}} Let $H$ be the Hadamard gate and $T_\phi$ a phase gate with an arbitrary phase $\phi$:
\begin{gather}
H=\frac{i}{\sqrt{2}}\left(\begin{array}{cc}
1 & 1 \\
1 & -1 \\
\end{array}\right),\,T_\phi=\left(\begin{array}{cc}
e^{-i\phi}& 0\\
0 & e^{i\phi} \\
\end{array}\right).
\end{gather}
Using our notation $H=U(\pi/2,\vec{k}_H)$, where $\vec{k}_H=\frac{1}{\sqrt{2}}(0,1,1)$ and $T_\phi=U(\phi,\vec{k}_{T_\phi})$, where $\vec{k}_{T_\phi}=(0,0,1)$, $\vec{k}_1\cdot\vec{k}_2=\frac{1}{\sqrt{2}}$.

Our goal is to check for which $\phi$, $\overline{<H,T_\phi>}=SU(2)$. 
\begin{description}
\item[Case 1]
If $\phi=k\pi$ then $T_\phi=\pm I$ and the generated group is the finite cyclic group of the order $4$ when $\phi=0$ or the order $8$ when $\phi=\pi$. 
\item[Case 2] 
When $\phi=\frac{k\pi}{2}$ and $k$ is odd, by Fact \ref{fact1} we have that $\mathcal{C}(\mathrm{Ad}_H,\mathrm{Ad}_{T_\phi})$ is larger than $\{\lambda I:\lambda\in\mathbb{R}\}$ and hence $\overline{<H,T_{\frac{k\pi}{2}}>}\neq SU(2)$. In fact it is the finite dicyclic group of order $16$ whose generators are $HT$ and $T$. Fixing universality in this case requires, for example, adding a matrix that has a non-exceptional spectrum and whose $\vec{k}$ is neither parallel nor orthogonal to $\vec{k}_H$ and $\vec{k}_{T_{\pi/2}}$. 
\item[Case 3]
For $\phi\neq\frac{k\pi}{2}$, again by Fact \ref{fact1}, $\mathcal{C}(\mathrm{Ad}_{H},\mathrm{Ad}_{T(\phi)})=\{\lambda I:\lambda\in\mathbb{R}\}$ and we just need to check if $<H,T_\phi>$ is infinite. We distinguish three possibilities:
 
{\noindent \textbf{1}.} We first assume that $\phi$ is not exceptional. Then by Theorem \ref{main} $\overline{{\langle H,T_\phi\rangle}}=SU(2)$. Our algorithm for deciding universality terminates at step 2 with $l=1$.

\noindent \textbf{2}. We next consider the exceptional angles. For 
\[
\phi\in \{\frac{k_3\pi}{3},\,\frac{k_5\pi}{5},\,\frac{k_6\pi}{6}\},\,\,\mathrm{gcd}(k_i,i)=1,
\]
we look at the product $U(\gamma,\vec{k}_{HT})=HT_\phi=U(\pi/2,\vec{k}_H)U(\phi,\vec{k}_T)$. Using formula (\ref{def:gamma}) we calculate $\cos\gamma$, compare it with $\cos\psi$ for all exceptional angles $\psi$ and find out they never agree. Hence $\gamma$ is not exceptional. Thus by Theorem \ref{main} we get $\overline{<HT_\phi>}=SU(2)$. Our algorithm for deciding universality terminates in Step 2 with $l=2$.

\noindent\textbf{3}. We are left with $\phi=\frac{k_4\pi}{4}$ where $\mathrm{gcd}(k_4,4)=1$. There are exactly four such angles. Calculations of $U(\gamma,\vec{k}_{HT_\phi})=HT_\phi$ shows that $\gamma$ is exceptional, i.e $\gamma=\frac{k_3\pi}{3}$, where $\mathrm{gcd}(k_3,3)=1$. Moreover, taking further products results in a finite subgroup consisting of $48$ elements (all have exceptional spectra) known as the binary octahedral group. Our algorithm for deciding universality terminates in Step 3 with $l=8$. Fixing non-universality can be accomplished by, for example, adding one gate $U(\psi,\vec{k}_\psi)$ with a non-exceptional $\psi$ and an arbitrary $\vec{k}_\psi$.
\end{description}

As we can see from the above example our algorithm requires at most words of length  $l=8$ to terminate for any $H$ and $T_\phi$.

\section{Acknowledgment} 
We would like to thank the anonymous referees for suggestions that led to improvements of the paper. This work was supported by National Science Centre, Poland under the grant SONATA BIS: 2015/18/E/ST1/00200.

\bibliographystyle{apsrev4-1}

\end{document}